\documentclass[a4paper]{article}

\usepackage[english]{babel} 
\usepackage{amsmath}
\usepackage{amsfonts}
\usepackage{amssymb}
\usepackage{amsthm,latexsym}
\usepackage{url}
\usepackage{enumerate}
\usepackage{graphicx}
\usepackage{fourier}
\usepackage{hyperref}
\usepackage[all]{xy}
\usepackage{color}
\usepackage{booktabs}

\usepackage{fullpage}
\theoremstyle{plain}
\newtheorem{thm}{Theorem}
\newtheorem{prop}[thm]{Proposition}
\newtheorem{cor}[thm]{Corollary}
\newtheorem{lem}[thm]{Lemma}

\theoremstyle{definition}
\newtheorem{defn}[thm]{Definition}

\newtheorem{nota}{Notation}

\theoremstyle{remark}
\newtheorem{rem}{Remark}

\newcommand{\F}{\mathbb{F}}
\newcommand{\Z}{\mathbb{Z}}
\newcommand{\Fqm}{\F_{q^m}}
\newcommand{\Fqmt}{\F_{q^{mt}}}
\newcommand{\Fqmr}{\F_{q^{mr}}}

\newcommand{\Fq}{\F_q}
\newcommand{\Fqq}{\F_{q^2}}
\newcommand{\A}{\Fqm^n}

\newcommand{\fract}[2]{\hbox{\leavevmode
\kern.1em \raise .5ex \hbox{\the\scriptfont0 $#1$}\kern-.1em }\big/
\hbox{\kern-.15em \lower .25ex \hbox{\the\scriptfont0 $#2$}}
}
\DeclareMathOperator{\tr}{Tr}
\newcommand{\NA}{\textrm{N}}
\newcommand{\TA}{\tr}

\newcommand{\ev}{\textrm{ev}}
\newcommand{\evL}[1]{\ev_L \left( #1 \right)}
\newcommand{\trF}[2]{\tr_{\ \! \! {\scriptstyle #1}/{\scriptstyle #2}}}
\newcommand{\map}[4]{\left\{
    \begin{array}{ccc}
      #1 & \longrightarrow & #2 \\
      #3 & \longmapsto     & #4
    \end{array}
\right.}
\newcommand{\gop}[2]{\Gamma_q\left(#1, #2\right)}

\newcommand{\eqdef}{\stackrel{\text{def}}{=}}

\newcommand{\bset}[2]{\left\{\left. #1 ~\right|~ #2 \right\}}

\renewcommand{\leq}{\leqslant}
\renewcommand{\geq}{\geqslant} 

\title{New Identities Relating Wild Goppa Codes}

\author{Alain Couvreur\thanks{GRACE Project --- INRIA Saclay \&  LIX, CNRS UMR 7161
--- \'Ecole Polytechnique, 91120 Palaiseau Cedex, France. 
\href{mailto:alain.couvreur@lix.polytechnique.fr}{alain.couvreur@lix.polytechnique.fr}},\ \
 Ayoub Otmani\thanks{Normandie Univ, France; 
UR, LITIS, F-76821 Mont-Saint-Aignan, France.
 \href{mailto:ayoub.otmani@univ-rouen.fr}{ayoub.otmani@univ-rouen.fr}}\ \
 and\ \ Jean--Pierre Tillich\thanks{SECRET Project --- INRIA Rocquencourt,   
78153 Le Chesnay Cedex, France. 
\href{mailto:jean-pierre.tillich@inria.fr}{jean-pierre.tillich@inria.fr}}}

\begin{document}

\maketitle

\begin{abstract}
For a given support 
$L\in \Fqm^n$
and a polynomial $g\in \Fqm[x]$
with no roots in $\Fqm$, we prove 
equality between the $q$--ary Goppa codes 
$\gop{L}{N(g)} = \gop{L}{N(g)/g}$ where $N(g)$ denotes the {\em norm}
of $g$, that is
$g^{q^{m-1}+\cdots +q+1}.$
In particular, for $m=2$, that is, for a quadratic extension,
we get $\gop{L}{g^q} = \gop{L}{g^{q+1}}$.
If $g$ has roots in
$\Fqm$, then 
we do not necessarily have equality 
and we prove that the difference
of the dimensions of the two codes 
is bounded 
above by the number of distinct roots of $g$ in $\Fqm$.
These identities provide numerous code equivalences and improved designed
parameters for some families of classical Goppa codes.
\end{abstract}



\section*{Introduction}
Let $\Fqm/\Fq$ be an extension of finite fields.
Given an ordered 
$n$--tuple $L=(\alpha_1, \ldots , \alpha_n) \in \Fqm^n$
and a polynomial $G\in \Fqm[x]$
with no roots among the entries of $L$,
the {\em classical Goppa code} over $\Fq$ denoted by $\gop{L}{G}$ 
is defined 
as
$$
\gop{L}{G} \eqdef \left\{(c_1, \ldots, c_n)\in \F_q^n \ ~\left|~ \sum_{i=1}^n \frac{c_i}{x-\alpha_i} \equiv 0 \mod G(x) \right. \right\}.
$$
Since their introduction by V.~D.~Goppa in 1970 \cite{ClassGoppa},
classical Goppa codes are subject to intense study and many questions remain
open.
For instance, even if the existence of asymptotic families of Goppa codes
reaching the Gilbert--Varshamov bound is known for a long time,
no explicit construction of such a family is known.
More generally, the exact computation of the dimension and the minimum distance
of a given Goppa code remain an open problem.

\smallskip

Besides, Goppa codes are particularly appealing for
cryptographic applications.
Indeed, since the introduction of code--based cryptography by 
McEliece in 1978 \cite{McEliece78}, Goppa codes still remain
among the few families of algebraic codes which resist to any structural
attack. This is one of the reasons why
every improvement of our knowledge of these codes
is of particular interest.

\smallskip

Goppa codes form a subfamily of {\em alternant codes}, that is subfield subcodes
of {\em Generalised Reed--Solomon codes}. As alternant codes, the classical
results on the parameters of subfield subcodes provide lower bounds for
their dimension and minimum distance.
However, these bounds can be improved for some specific
Goppa codes and for a relevant choice of the Goppa 
polynomial $G$.
A major improvement of these parameters has been obtained in 1976 by Sugiyama
{\it et al.} \cite{Sugiyama_et_al} who proved that
if $g\in \Fqm [x]$ is squarefree, then,
$\gop{L}{g^{q-1}} = \gop{L}{ g^{q}}$. This equality can easily be generalised as
$\gop{L}{g^{sq-1}} = \gop{L}{g^{sq}}$ for any positive integer $s$.
This identity relating the subfield subcodes of two Generalised Reed--Solomon
codes with distinct parameters allows to take the best from each one.
Namely, the dimension of such a code is at least the designed dimension of
$\gop{L}{g^{q-1}}$ which is $n - m \deg(g) (q-1)$ and the minimum
distance is at least the designed distance of $\gop{L}{g^{q}}$ which
equals $\deg(g)q+1$.
In the binary case, this identity provides a lower bound for the minimum
distance of $\Gamma_2 (L, g)$ which is almost twice the designed distance
for alternant codes.
Some extensions of Sugiyama {\it et al.}'s result to algebraic geometry codes
are presented in \cite{oimCartier, KatsmanTsfasman, Wirtz}.
The particular subclass of Goppa codes of the form
$\gop{L}{g^{q-1}}$ for a squarefree Goppa polynomial $g$ has been called
{\em wild Goppa codes}
by Bernstein {\it et al.}~\cite{BLP_WildMcE, BLP_WildMcE_Incognito} who
proposed them
for McEliece's encryption scheme since their improved designed
parameters allowed to reduce the size of the public and secret keys for a
fixed security level.

Beside Sugiyama {\it et al.}'s results, many
improved lower bounds and exact computations
of the true parameters --- in particular the dimension --- of some particular
Goppa codes appear in the literature.
For instance (and the list is far from being exhaustive),
the authors of \cite{LoelianConan} propose a new lower bound for the
minimum distance using the discrete Fourier transform.
Improved lower bounds or exact values of the dimension of Goppa codes 
for specific families of Goppa polynomials are proved in
\cite{BezzShekDCC,  RHA, vdV, Veron, VeronDCCTraces, VeronDCCproofs}.
Many code equivalences and inclusions relating some particular binary Goppa
codes are proved in \cite{BezzShekIEEE,BezzShekArX}.
Most of these results concern Goppa codes whose Goppa polynomial or 
one of its divisors sends every entry of the support $L \in \Fqm^n$
 into a proper subfield of $\Fqm$.
 Such a feature 
induces in general the apparition of linear relations
 between the parity checks of the code when passing from the Generalised
 Reed--Solomon code to its subfield subcode, which guarantees a larger dimension
 compared to the generic estimate for subfield subcodes.
 Among the previously cited works we should point out V\'eron's examples \cite{Veron}, who studied
  Goppa codes whose Goppa polynomial
 is a {\em trace polynomial}, {\it i.e.} a polynomial $G \in \Fqm [x]$
 of the form $G= g + g^q + \cdots + g^{q^{m-1}}$ with $g\in \Fqm[x]$.
 Roughly speaking, the present article, deals with norms instead of traces.
 Namely, we consider Goppa polynomials of the form
 $g^{q^{m-1}+\cdots  +q+1}$ and prove a very 
 surprising equality: for 
 $g\in \Fqm[x]$ with no roots in $\Fqm$, we have:
 $$
 \gop{L}{g^{q^{m-1}+\cdots  +q}} = \gop{L}{g^{q^{m-1}+\cdots  +q+1}}.
 $$
 To the best of our knowledge, this article provides the first
 new general identity relating Goppa codes since 
Sugiyama {\it et al.}'s article.

 \subsection*{Results of the present article}
 Consider an extension of finite fields $\Fqm/\F_q$ with $m\geq 2$.
 Let $n$ be a positive integer and $L = (\alpha_1, \ldots, \alpha_n)$
 be an ordered $n$--tuple of pairwise
 distinct elements of $\Fqm$ and $G\in \Fqm [x]$ be a polynomial with no roots
 among the entries of $L$,
 then the classical Goppa code associated to $L$ and $G$
 over the subfield $\F_q$ is defined as:
 $$
 \gop{L}{G} \eqdef \left\{(c_1, \ldots, c_n)\in \F_q^n \ ~\left|~ \sum_{i=1}^n \frac{c_i}{x-\alpha_i} \equiv 0 \mod G(x) \right. \right\}.
 $$
 The $n$--tuple $L$ is called the {\em support} of the code.
 If $L$ contains every element of $\Fqm$, {\it i.e.} if $n = q^m$, then
 the corresponding codes are said to have a {\em full support}.
 The polynomial $G$ is called the {\em Goppa polynomial}.
 As an alternant code
 a Goppa code has a {\em designed dimension} $n-m\deg(G)$
 and a {\em designed minimum distance} $\deg(G)+1$
 (see \cite[Theorem 9.2.7]{vanLint}).
 Here we state Theorems~\ref{thm:equality} and~\ref{thm:eqbis}
 which are the main results of the present article.
 Their proofs are given in Section~\ref{sec:mainproof}.

 \begin{thm}\label{thm:equality}
   Let $g\in \Fqm [x]$ be a polynomial
   with no roots in $\Fqm$
   and $L$ be an ordered $n$--tuple of pairwise distinct elements
   of $\Fqm$.
   Then,
   \begin{equation}
     \label{eq:thmfirst}
     \gop{L}{ g^{q^{m-1} + q^{m-2}+\cdots +q}} =
     \gop{L}{ g^{q^{m-1} + q^{m-2}+\cdots +q+1}}.
   \end{equation}
 \end{thm}

 This result can be combined with Sugyiama
 {\it et al.}~\cite{Sugiyama_et_al} and 
 gives the following corollary.

 \begin{cor}
     Let $L, g$ be as in Theorem~\ref{thm:equality} and 
     assume in addition that $g$ is squarefree, then
   \begin{equation}
   \label{eq:thmsec}
     \gop{L}{ g^{q^{m-1} + q^{m-2}+\cdots +q -1}} =
     \gop{L}{ g^{q^{m-1} + q^{m-2}+\cdots +q}} =
     \gop{L}{ g^{q^{m-1} + q^{m-2}+\cdots +q+1}}.
   \end{equation}
 \end{cor}


 In addition, Theorem~\ref{thm:equality} provides improved
 designed parameters for the involved codes,
 namely, they can easily be proved to have parameters of the form:
 $$
 [n ,\ \geq n- mt(q^{m-1}+\cdots +q -1),\ \geq t(q^{m-1}+\cdots +q +1) +1],
 $$
 where $t$ denotes the degree of $g$.
 Actually, the dimension is far larger than this bound.
 Indeed, the polynomial $g^{q^{m-1}+\cdots +q+1}$
 sends every element $\alpha \in \Fqm$ on an element of $\Fq$,
 namely the norm of $g(\alpha)$. 
 In \cite{HMO'S}, the authors  prove that such alternant codes are 
 equivalent to a subfield subcode of a Reed--Solomon code, that is extended
 or shortened BCH codes.
 Furthermore, it is well--known that subfield subcodes of Reed--Solomon codes
 have a large dimension compared to subfield subcodes of random codes
 \cite{BierbrauerEdel, HMcS98, HMO'S, Liao2011}.
 For instance when $m=2$, the codes $\gop{L}{g^q}$ and
 $\gop{L}{g^{q+1}}$ are equal and have parameters of the form:
 $$
 [n,\ \geq n- 2t (q-1) + t(t-2) ,\ \geq t(q+1)+1].
 $$
 Third, we point out that compared to Sugyiama et. al.'s result \cite{Sugiyama_et_al},
 our identity (\ref{eq:thmfirst}) does not require the polynomial $g$ to be squarefree.
 This has the following interesting consequence.
 \begin{cor}\label{cor:main}
   Let $h$ be a polynomial in $\Fqm[x]$ with no roots in
   $\Fqm$ and $L$ be a support. Then, for all integer $s>0$, we have
   $$
   \gop{L}{h^{s(q^{m-1}+q^{m-2}+\cdots +q)}}
   = \gop{L}{h^{s(q^{m-1}+q^{m-2}+\cdots +q+1)}}
   $$
   and all the intermediary codes $\gop{L}{h^{s(q^{m-1}+q^{m-2}+\cdots +q)+i}}$
   for $0<i<s$ are also equal to the above codes.
 \end{cor}

 This corollary can be also combined with Sugyiama {\it et al.}'s result
 assuming that the polynomial $h$ is squarefree, which will extend the equality
 as:
 $$
   \gop{L}{h^{s(q^{m-1}+q^{m-2}+\cdots +q)-1}}
   = \cdots  = \gop{L}{h^{s(q^{m-1}+q^{m-2}+\cdots +q+1)}}
 $$

 Finally, it is worth noting that, even if the code has not a full support,
 Theorem~\ref{thm:equality} holds true only if $g$ has no
 roots in $\Fqm$. In particular, this result is not 
 usable when the degree of $g$ is $1$.
 Nevertheless, in the general case one still has the following statement.

 \begin{thm}
   \label{thm:eqbis}
   Let $L$ be a support and $g\in \Fqm[x]$ be a polynomial with no
   roots in $L$.
   Let $r$ be the number of distinct roots of $g$ ({\it i.e.} not counted with
   multiplicity) in $\Fqm$. Then we have:
   $$
   \dim_{\Fq} \gop{L}{g^{q^{m-1}+q^{m-2}+\cdots + q}} - \dim_{\Fq}
   \gop{L}{g^{q^{m-1}+q^{m-2}+\cdots + q+1}} ~\leq~ r.
   $$
 \end{thm}

 Notice that in general the difference between the dimensions of $\gop{L}{g^a}$
 and $\gop{L}{g^{a+1}}$ is $m\deg (g)$. Here the difference is smaller
 than $\deg(g)$
 and is not multiplied by $m$. Thus, the difference is small compared to
 the general case.
 This statement is of interest, since, using the very same argument
 as above, one can prove using \cite{HMO'S} that, if the support $L$ is full
 ({\it i.e.} $n = q^m$), then
 $\gop{L}{g^{q^{m-1}+q^{m-2}+\cdots + q+1}}$ is a subfield subcode of a 
 Reed--Solomon code and hence has a dimension larger than the designed dimension
 for general alternant codes.
 By this manner, Theorem~\ref{thm:eqbis} provides an improved lower bound for
 the dimension of codes $\gop{L}{g^{q^{m-1}+q^{m-2}+\cdots + q}}$, where $g$ has degree $1$.


 \subsection*{Outline of the article}
 This article is organised as follows.
 Elementary properties of Goppa codes are recalled and discussed in
 Section~\ref{sec:props}. The particular case of Goppa codes whose Goppa
 polynomial sends every entry of $L$ into a proper subfield of $\Fqm$
 is discussed
 in Section~\ref{sec:proof1}. Section~\ref{sec:mainproof} is devoted to the 
 proofs of the main results of the article, namely Theorems~\ref{thm:equality} 
 and \ref{thm:eqbis}.
 Finally, some numerical examples illustrating our results are presented
 in Section~\ref{sec:examples}.

 \section{Some Well Known Properties of Goppa Codes}
 \label{sec:props}



 \begin{nota}
  For a given support $L = (\alpha_1, \ldots , \alpha_n)$ of pairwise
  distinct elements of $\Fqm$,
  we denote by $\pi_L$ the polynomial
  $$\pi_L\eqdef \prod_{i=1}^n (x-\alpha_i).$$
  We denote by $\pi_L'$ its first derivative.
  Finally, for a positive integer $a$, we denote by $\Fqm[x]_{<a}$
  the subspace of $\Fqm[x]$ of polynomials of degree less than $a$.
 \end{nota}


 Recall that $q$-ary Goppa codes are alternant codes, {\it i.e.} subfield subcodes over $\F_q$ of a Generalised Reed--Solomon (GRS) code over
 $\Fqm$. Therefore, 
 from Delsarte's Theorem \cite[Theorem 7.7.11]{slmc1}, the dual of the Goppa code, is the trace of a GRS
 code.

 \begin{lem}\label{lem:evgoppa}
   Let $L$ be a support and $h\in \Fqm [x]$
   with no roots in $L$
   and $C, C^{\bot}$ be the GRS codes defined by
   \begin{align*}
   C  \eqdef & \left\{ \left. \left(\frac{h (\alpha_1) f(\alpha_1)}{\pi'_L(\alpha_1)},\ldots , 
     \frac{h (\alpha_n) f(\alpha_n)}{\pi'_L(\alpha_n)}\right)\ \right| \ 
     \alpha_i \in L,\ f\in \F_{q^m}[x]_{<n-t}   \right\}; \\
   C^{\bot}  = &
   \left\{ \left. \left(\frac{ f(\alpha_1)}{h(\alpha_1)},\ldots ,
     \frac{ f(\alpha_n)}{h(\alpha_n)}\right)\ \right| \ 
     \alpha_i \in L,\ f\in \F_{q^m}[x]_{<t}  \right\}.
   \end{align*}
   Then, $\gop{L}{ h} = C_{|\Fq}$ and $\gop{L}{ h}^{\bot} =
   \tr (C^{\bot})$ where $\tr$ denotes the map
   $\trF{\Fqm}{\Fq}$ applied component-wise.
 \end{lem}

 \begin{proof}
  See \cite[Theorems 12.4 and 12.5]{slmc1} for binary Goppa codes.
  The $q$--ary case is obtained using the very same proof.
 \end{proof}


 The following elementary lemma is useful
 in what follows.


 \begin{lem}\label{lem:split}
   Let $L = (\alpha_1, \ldots , \alpha_n)$ be a support and
   $g,h \in \Fqm[x]$ be two relatively prime polynomials such that
   both have no roots among the entries of $L$.
   Set $n-a = \dim_{\Fq} \gop{L}{g}$ and $n-b = \dim_{\Fq} \gop{L}{h}$  
   Then,
   \begin{align}
     \label{eq:intersect}
     \gop{L}{gh} & =  \gop{L}{g} \cap \gop{L}{h};\\
     \label{eq:dim_prod}
     \dim_{\Fq} \gop{L}{gh} &\geq  n-a-b.
   \end{align}
 \end{lem}

 \begin{proof}
   The Chinese remainder Theorem in the ring
   $\Fqm[x,\frac{1}{\pi_L}]$
   asserts that
   $$
   \sum_{i=1}^n \frac{c_i}{x-\alpha_i} \equiv 0 \mod (gh) \quad
   \Longleftrightarrow \quad
   \left\{ \begin{aligned}
   \sum_{i=1}^n \frac{c_i}{x-\alpha_i} &\equiv 0 \mod (g)\\
   \sum_{i=1}^n \frac{c_i}{x-\alpha_i} &\equiv 0 \mod (h).
   \end{aligned}
   \right.
   $$
   This yields~(\ref{eq:intersect})
   and implies that
   $\gop{L}{gh}^{\bot} = \gop{L}{g}^{\bot} + \gop{L}{h}^{\bot}$,
   which gives (\ref{eq:dim_prod}).
 \end{proof}

 \begin{rem}
   Let $L , g$ be as in Theorem~\ref{thm:equality} and
   $h\in \Fqm[x]$ be a polynomial prime to $g$
   and with no roots in $L$. Then, Theorem~\ref{thm:equality}
   generalises as:
   $$
   \gop{L}{hg^{q^{m-1}+\cdots +q-1}} =  \gop{L}{hg^{q^{m-1}+\cdots +q}}
   =  \gop{L}{hg^{q^{m-1+\cdots +q+1}}}.
   $$
   The Goppa codes described above for $m=2$
   are proposed for cryptographic applications
   in \cite{BLP_WildMcE_Incognito}.
 \end{rem} 


 \begin{lem}\label{lem:supports}
   Let $L, L'$ be two supports such that $L$ can be obtained from
   $L'$ by removing some entries without changing the ordering.
   Let $g\in \Fqm[x]$ be a polynomial with no roots in $L'$ (and hence in $L$),
   then $\gop{L}{g}$ is equal to the shortening of $\gop{L'}{g}$
   on $L$.
 \end{lem}

 \begin{proof}
   This follows immediately by viewing
   $
   \sum_{\alpha \in L} \frac{c_{\alpha}}{x-\alpha}
   $
   as the sum
   $
   \sum_{\beta \in L'} \frac{c_{\beta}}{x-\beta}$ 
   such that $c_{\beta}=0$ for all $\beta \in L'\setminus L$.
 \end{proof}

 \section{Goppa Codes and Subfield Subcodes of Reed--Solomon Codes}\label{sec:proof1}

 \begin{defn}[Diagonal equivalence]\label{def:streq}
   Let $C, C'$ be two codes in $\Fq^n$.
   $C, C'$ are 
   said to be {\em diagonally equivalent}
   and 
   we write 
   $C\sim_{\Fq} C'$ in this case,
   if $C'$ is the image of $C$ by a Hamming isometry of $\Fq^n$ of the form:
   $$
   \map{\Fq^n}{\Fq^n}{(x_1, \ldots , x_n)}{(u_1 x_1, \ldots , u_n x_n)},
   $$
   where the $u_i$'s are all in $\Fq^{\times}$. 
   This choice of terminology comes from the fact that the codes can be sent
   onto each other using 
   an invertible diagonal matrix.
 \end{defn}
 Here, we reformulate for our purpose some results stated in \cite{HMO'S}.

 \begin{prop}\label{prop:equiv}
   Let $L$ be an $n$--tuple of pairwise distinct elements of 
   $\Fqm$ and $g,h$ be two polynomials in $\Fqm[x]$ which have no roots 
   among the entries of $L$
   (but possibly elsewhere in 
   $\Fqm$),
   such that $\deg(g)=\deg(h)$,
   then the codes $\gop{L}{g^{q^{m-1}+\cdots +q+1}}$ and
   $\gop{L}{h^{q^{m-1}+\cdots +q+1}}$ are diagonally equivalent.
 \end{prop}

 \begin{proof}
   For all $\alpha \in L$, 
   note that
   $g^{q^{m-1}+\cdots +q+1}(\alpha)$ (resp. $h^{q^{m-1}+\cdots +q+1}(\alpha)$)
   is nothing but $\textrm{N}_{\Fqm/\Fq}(g(\alpha))$
   (resp. $\textrm{N}_{\Fqm/\Fq}(h(\alpha))$) and hence is in $\Fq$.
   Thus both Goppa polynomials send every entry of $L$ into $\Fq$.
   One concludes
  using \cite[Proposition 3.5]{HMO'S}.
\end{proof}

\begin{rem}
  The result stated in \cite{HMO'S} concerns codes on a support avoiding $0$.
  However, their result extends straightforwardly to a full support. 
\end{rem}

\begin{cor}\label{cor:equivalence}
  Let $g$ be a polynomial in $\Fqm [x]$ with no roots in $\Fqm$
  and $L_0$ be a ``full support'', {\it i.e.}
  an ordered $q^m$--tuple containing all the elements
  of $\Fqm$, then we have the diagonal equivalence of codes
  $$
    \gop{L_0}{ g^{q^{m-1} + q^{m-2}+\cdots +q+1}}
    \sim_{\Fq} RS_{k}(L_0)_{|\F_q},
  $$
  where $RS_{k}(L_0)_{|\F_q}$
  denotes the subfield subcode 
  of the Reed--Solomon code over $\Fqm$ of dimension
  $k= q^m- \deg(g)(q^{m-1}+q^{m-2}+\cdots +q+1)$ with full support.
  In particular the diagonal equivalence class of this code depends only
  on the degree of $g$.
\end{cor}

\begin{proof}
See \cite[\S 3]{HMO'S}.
\end{proof}



It is worth noting that a subfield subcode of a full support Reed--Solomon
code is nothing but an extended BCH code.
In \cite[Theorem 4.4]{HMO'S}, the authors give a formula for the dimension
of such codes involving the number and the size of some cyclotomic classes.
See Section~\ref{sec:examples} for further discussion.

\begin{rem}
  Corollary~\ref{cor:equivalence} holds for every Goppa Polynomial
  sending every entry of the support into $\Fq$.
  In particular, this gives another interpretation of 
  V\'eron's results~\cite{Veron} showing
  that the dimension of Goppa codes with a Goppa polynomial
  of the form $g+g^q+\cdots + g^{q^{m-1}}$ exceeds the
  generic bound for alternant codes.
  Indeed, since
  such a Goppa polynomial sends every entry of the support into $\Fq$,
  the 
  corresponding Goppa code is $\Fq$--equivalent to a BCH code.
\end{rem}

\section{Proof of Theorems~\ref{thm:equality}
and~\ref{thm:eqbis}}
\label{sec:mainproof}

\subsection{Notation}\label{subsec:nota}
  In what follows we frequently consider 
  $\Fqm^n$ and $\Fq^n$
  as rings for their canonical product ring structure.
  The component-wise product of two $n$--tuples $a, b$ in $\Fqm^n$
  is denoted by
  $$a\star b \eqdef (a_1 b_1, \ldots , a_n b_n).$$
  We also allow ourselves the notation 
  $a^s$ to denote the component-wise $s$--th power 
  and $1/a$ for the component-wise inverse when $a$ is in
  $(\Fqm^{\times})^n$.
  Recall that 
$\TA : \Fqm^n \rightarrow \Fq^n$
  is the component-wise trace map. 
  In the same manner, 
$\NA : \Fqm^n \rightarrow \Fq^n$ is the component-wise norm map.
  Furthermore, we denote by $\ev_L$ the 
  \emph{evaluation} function:
  \begin{equation}
    \label{eq:evL}
        \ev_L : \map{\Fqm[x]}{\A}{f}{\left(f(\alpha_1), \ldots , f(\alpha_n)\right)}.
  \end{equation}
  This turns out to be a ring homomorphism. 
  We also need to introduce the map $\tau$ defined as the 
  composition of $\ev_L$ and $\TA$, namely:
  \begin{equation}
    \label{eq:eta}
   \tau : \map{\Fqm[x]}{\Fq^n}{f}{\TA (\ev_L(f))}.
  \end{equation}
  
   Finally, for convenience, we denote by $e$ the integer
    $$
      e\eqdef q^{m-1} + q^{m-2} +\cdots + q.
    $$

\subsection{Preliminaries} \label{subsec:outline}


\subsubsection{Local reformulation}
Thanks to Lemma~\ref{lem:split}, 
one can assume that
$g$ is a power $h^s$ 
of an irreducible polynomial
$h \in \Fqm[x]$ with $s \geq 1$.
Hence we are reduced to prove the following statement.
\begin{thm}[Local version of Theorems~\ref{thm:equality} and \ref{thm:eqbis}]\label{thm:irred}
  Let $L\in \Fqm^n$ be a support and $g\in \Fqm[x]$ be a
  polynomial of degree $t$
  which is either irreducible or a power of an irreducible polynomial.
   Only two cases can occur:
  \begin{enumerate}[(i)]
  \item\label{item:thmeq} $g$ has no roots in $\Fqm$
    then $\gop{L}{g^{e}} = \gop{L}{g^{e+1}}$;
  \item\label{item:thmeqbis} or $g = (x- \rho)^t$  
    for some $t\geq 1$ and some $\rho \in \Fqm$ which
    does not appear in the entries of $L$,
    then: 
\[
\dim_{\Fq} \gop{L}{g^{e}} - \dim_{\Fq}\gop{L}{g^{e+1}} \leq 1.
\] 
  \end{enumerate}
\end{thm}

\subsubsection{Duality and role of the norm}\label{subsec:dualRef}

Theorem~\ref{thm:irred} can be reformulated using duality together with an
argument involving the norm
$\NA_{\Fqm/\Fq}$
which is strongly related to the results of
Section~\ref{sec:proof1}.
\begin{prop}
  \label{prop:dualRef}
Theorem~\ref{thm:irred}
(\ref{item:thmeq}) and (\ref{item:thmeqbis}) are respectively equivalent to
\begin{enumerate}[{\it (i')}]
\item\label{it:ip} If $g$ has no roots in $\Fqm$ then
$\displaystyle
  \tau\left( \Fqm[x]_{<(e+1)t} \right)
   = \tau\left( g\Fqm[x]_{<et} \right)$.
\item\label{it:iip} If $g = (x-\rho)^t$ for some $\rho \in \Fqm$ and some $t\geq 1$,
  then
$$
\dim_{\Fq}  \tau \left( \Fqm[x]_{<(e+1)t}\right)
   -\dim_{\Fq} \tau \left( g\Fqm[x]_{<et}\right) \leq 1.
$$
\end{enumerate}
\end{prop}

\begin{proof}
Let us prove that Theorem~\ref{thm:irred}~(\ref{item:thmeq})
is equivalent to $(i')$.
Let $g$ be a polynomial with no roots in $\Fqm$.
By using  Lemma~\ref{lem:evgoppa}, Theorem~\ref{thm:irred}~(\ref{item:thmeq})
is equivalent to its dual reformulation, namely
\begin{equation}
  \label{eq:dualR}
 \bset{ \TA \left( \frac{\evL{f} }{\evL{g^{e+1}}} \right)}{f\in \Fqm[x]_{<(e+1)t}}
= \bset{ \TA \left( \frac{\evL{f} }{\evL{g^e}} \right)}{f\in \Fqm[x]_{<et}}
\end{equation}
Note that
\begin{eqnarray}
\bset{ \TA \left( \frac{\evL{f} }{\evL{g^e}} \right)}{f\in \Fqm[x]_{<et}} & = & \bset{ \TA \left( \frac{\evL{f}\star \evL{g} }{\evL{g}^e \star \evL{g}} \right)}{f\in \Fqm[x]_{<et}}\\
& = &  \bset{ \TA \left( \frac{\evL{h} }{\evL{g}^{e+1}} \right)}{h\in g\Fqm[x]_{<et}}.
\end{eqnarray}
In addition, since 
$\evL{g^{e+1}}= 
 \NA\left(\evL{g}\right)$, this vector
has its entries in $\Fq$. Hence these denominators can be pulled out of the 
traces. Therefore, 
\eqref{eq:dualR}  is equivalent to
\begin{equation}
  \label{eq:bidule}
   \bset{ \frac{1}{\evL{g}^{e+1}} \star \TA \left( \evL{f}  \right)}{f\in \Fqm[x]_{<(e+1)t}}
   =  \bset{ \frac{1}{\evL{g}^{e+1}} \star \TA \left( \evL{h}  \right)}{h\in g\Fqm[x]_{<et}}.
\end{equation}
Finally, this equality is clearly equivalent to
\begin{equation}
  \label{eq:chose}
   \bset{  \TA \left( \evL{f}  \right)}{f\in \Fqm[x]_{<(e+1)t}}
   =  \bset{  \TA \left( \evL{h}  \right)}{h\in g\Fqm[x]_{<et}}.
\end{equation}
This concludes the proof.
The proof of the equivalence between
Theorem~\ref{thm:irred}~(\ref{item:thmeqbis})
and {\it (ii')} is the very same one replacing equalities by inclusions
with codimension $1$.
\end{proof}

\subsubsection{The spaces $K$ and $T$}

Proposition \ref{prop:K} below explains that the proof of Proposition~\ref{prop:dualRef}, which is equivalent to Theorem~\ref{thm:irred}, can be achieved 
by proving the existence of two special vector spaces $K$ and $T$.

\begin{prop}\label{prop:K}
  If there exists a subspace $K$ of 
  $\Fqm[x]_{<(e+1)t}$ satisfying:
\begin{enumerate}[(I)]
\item\label{item:subset} $K\subset \ker
\tau$
\item\label{item:cap} $K \cap g\Fqm[x]_{<et} = \{0\}$
\item\label{item:dimK0} $\dim_{\Fq} K = mt-1$
\end{enumerate}
then it implies: 
$$\dim_{\Fq}  \tau \left( \Fqm[x]_{<(e+1)t}\right)
   -\dim_{\Fq} \tau \left( g\Fqm[x]_{<et}\right) \leq 1.
$$

\medskip

In addition, if $g$ has no roots in $\Fqm$,
and if there exists another $\Fq$subspace $T$ of 
$\Fqm[x]_{<(e+1)t}$ such that
\begin{enumerate}[(I)]
\setcounter{enumi}{3}
\item\label{item:K} $K \oplus T \subset \ker \tau$ \ \ \ and \ \ \ 
$\Fqm[x]_{<(e+1)t} = K \oplus  T \oplus g\Fqm[x]_{<et} $,
\end{enumerate}
then the equality $  \tau\left( \Fqm[x]_{<(e+1)t} \right)
   = \tau\left( g\Fqm[x]_{<et} \right)$ holds.
\end{prop}

\begin{proof}
This follows basically from the fact that $\dim_{\Fq}   \left( \Fqm[x]_{<(e+1)t}\right) = m(e+1)t$ whereas
$ \dim_{\Fq}   \left( \Fqm[x]_{<et}\right) = met$.
 If such a space $K$ exists, then
from (\ref{item:cap}) and (\ref{item:dimK0}), we get the existence
 of an $\Fq$--one--dimen\-sional subspace $T_0$ of $\Fqm[x]_{<(e+1)t}$ such that
 $ \Fqm[x]_{<(e+1)t} = K \oplus T_0 \oplus g\Fqm[x]_{<et}. $
Then, (\ref{item:subset}) leads to Proposition~\ref{prop:dualRef}~(\ref{it:iip}')
 or equivalently Theorem~\ref{thm:irred}~(\ref{item:thmeqbis}).  
 If in addition, there exists a space $T$ satisfying~(\ref{item:K}),
 then we clearly get Proposition~\ref{prop:dualRef}~(\ref{it:ip}')
 or equivalently Theorem~\ref{thm:irred}~(\ref{item:thmeq}).
 \end{proof}



\bigskip

From now on, the polynomial $g$ is assumed to be either irreducible
or of the form $g= h^s$ for some irreducible polynomial $h\in \Fqm[x]$ and some
integer $s>1$.
The degree of $g$ is denoted by $t$.

\subsection{The construction of $K$ and the existence of $T$}\label{subsec:K0}
 First, we recall a well known result concerning elements whose trace is zero (see \cite[Theorem 2.25]{LidlNiederreiter} for a proof).
\begin{lem}\label{lem:H90}
  For all $\alpha\in \Fqm$ such that $\tr_{\Fqm/\Fq}(\alpha)=0$, there exists
  $\beta \in \Fqm$ such that $\alpha = \beta - \beta^q$.
\end{lem}

We look for an $\Fq$--subspace
$K \subset \Fqm[x]_{<(e+1)t}$ satisfying~(\ref{item:subset}),
(\ref{item:cap}) and (\ref{item:dimK0}) of Proposition~\ref{prop:K}.
Notice that we have
$
\ker \tau  =  \left\{ a \in \Fqm[x]_{<(e+1)t}\; |\; \TA (\evL{a})=0\right\}.
$
This point, together with  Lemma~\ref{lem:H90}, explain the rationale behind the following definition.
  
\begin{defn}\label{def:K0}
  We denote by $K$ the image of the map
  \begin{equation}\label{eq:defmu}
  \mu : \map{\Fqm[x]_{<t}}{\Fqm[x]}{a}{a^q-a}.
  \end{equation}
\end{defn}

\begin{lem}\label{lem:K0}
  We have, $K \subset \ker \tau$ and
  $\dim_{\Fq} K = mt-1.$
\end{lem}

\begin{proof}
  First, let us check that $K \subset \Fqm[x]_{<(e+1)t}$.
  Let $a$ be an element of $\Fqm[x]_{<t}$,
  we have $a^q-a \in \Fqm[x]_{<qt}$.
  Since, $e=q^{m-1}+\cdots + q$ and $m\geq 2$,
  we have $e\geq q$ and hence $a^q-a \in \Fqm[x]_{<(e+1)t}$.
  Next, by definition of $\TA$ and 
  its elementary properties,
  we have:
  $$
  \TA \left(\evL{a^q - a} \right) =
  \TA \left(\evL{a}^q\right) - \TA \left( \evL{a}\right)=
  0.
  $$
This yields the inclusion $K \subset\ker \tau$.
  To get the dimension, we prove that $\dim_{\Fq} \ker \mu = 1$.
  Let $a$ be an element of $\ker \mu$, {\it i.e.} $a$ is a polynomial
  in $\Fqm[x]_{<t}$ satisfying $a(x)^q= a(x)$. Then
  the degree of $a$ is zero and $a$ is nothing but a constant polynomial
  satisfying $a^q=a$.
  Thus, $\ker \mu$ consists in the subspace of constant polynomials
  lying in $\Fq$.
\end{proof}

  In what follows some proofs require the use of congruences modulo
some polynomials. For this reason we introduce the following notation.

\begin{nota}\label{nota:class}
For $f \in \Fqm[x]$ and for all $a \in \Fqm[x]$,
we denote by $a\mod (f)$ the class of $a$ in $\fract{\Fqm[x]}{(f)}$.
In the same manner, the image of $K$ by the canonical map
$\Fqm[x] \rightarrow \fract{\Fqm[x]}{(f)}$ will be denoted by $(K \mod f)$.
Finally, after the map $\mu$ introduced in~(\ref{eq:defmu})
we define for all $f\in \Fqm[x]$, the map $\mu_f$ as
  \begin{equation}\label{eq:muf}
  \mu_f :
  \map{\fract{\Fqm[x]}{(f)}}{\fract{\Fqm[x]}{(f)}}{{a\mod (f)}}{{a^q-a \mod (f)}}.
  \end{equation}
\end{nota}

\begin{lem}\label{lem:kerTr}
  Let $h$ be an irreducible polynomial of degree $r$
  such that $g=h^s$ for some positive integer
  $s$ 
  (possibly $s = 1$).
  Viewing $\fract{\Fqm[x]}{(h)}$ as the finite field
  $\Fqmr$ 
  it turns out that  $K \mod (h)$ satisfies:
  $$
    K \mod (h) =  \ker \trF{\Fqmr}{\Fq}.
  $$
\end{lem}
\begin{rem}
  In particular, if $g$ is irreducible ($s$=1) 
  then $K \mod (g) =  \ker \trF{\Fqmt}{\Fq}.$
\end{rem}

\begin{proof}
  Since $r\eqdef \deg(h)\leq t$, the map 
  $\Fqm[x]_{<t} \rightarrow \fract{\Fqm[x]}{(h)}$
  is surjective
  and hence $(K \mod (h))$
  is nothing but the image
  of $\mu_h$ (see (\ref{eq:muf}) in Notation~\ref{nota:class}).
  Considering the quotient ring $\fract{\Fqm[x]}{(h)}$
  as the finite field $\Fqmr$ we conclude 
  by 
  using Lemma~\ref{lem:H90}.
\end{proof}

\begin{prop}
  \label{prop:dimKbar}
  We have
  $$
  \dim_{\Fq} (K \mod (g)) = mt-1.
  $$
\end{prop}

\begin{proof}
  If $g$ is irreducible, then it is a straightforward consequence of
  Lemma~\ref{lem:kerTr}.
  Let us assume that $g$ is of the form $h^s$ for some irreducible
  polynomial $h$ and some integer $s>1$.
  Recall that the space $K \mod (g)$
  is nothing but the image of $\mu_g$ (see (\ref{eq:muf}) in
  Notation~\ref{nota:class}).
  Therefore, we wish to prove that $\dim_{\Fq} \ker \mu_g=1$.
  We will show that $\ker \mu_g$ is isomorphic to $\Fq$.

  Let $a\in \Fqm[x]$ such that $(a \mod (g)) \in \ker \mu_g$. That is
  \begin{equation}
    \label{eq:amodg}
      a\equiv a^q \mod (g)  
  \end{equation}
  Since $g = h^s$, we have {\it a fortiori}
  $a \equiv a^q \mod (h)$. Since $\fract{\Fqm[x]}{(h)}$ is a field
  containing $\Fq$, then $a\mod (h)$ is represented by a constant
  polynomial lying in $\Fq$. Therefore, there exists $\alpha \in \Fq$
  and $a_1(x) \in \Fqm[x]$, such that
  $$
  a(x) = \alpha + h(x)a_1(x).
  $$
  From (\ref{eq:amodg}), we get $a\equiv a^{q^i} \mod (g)$ for all $i>0$.
  Choose $i$ such that $q^i \geq s.$ Then, $h^{q^i} \equiv 0 \mod (g)$ since
  $g=h^s$. Therefore, $a \equiv a^{q^i} \mod (g)$ entails
  $$
  a \equiv \alpha^{q^i} \mod (g).
  $$
  Finally, since $\alpha \in \Fq$, we have $\alpha^{q^i} = \alpha$
  which entails $a\equiv \alpha \mod (g)$.
  This yields an $\Fq$--isomorphism between $\ker \mu_g$ and $\Fq$, which
  concludes the proof.
\end{proof}

\begin{cor}\label{cor:cap}
  $K \cap g\Fqm[x]_{<et} = \{0\}$.
\end{cor}

\begin{proof}
  From Lemma~\ref{lem:K0}, the space $K$
  has $\Fq$--dimension $mt-1$ and, from Proposition~\ref{prop:dimKbar},
  $K \mod (g)$ has $\Fq$--dimension $mt-1$ too. 
  Thus, the canonical projection $K \rightarrow \fract{\Fqm[x]}{(g)}$
  is injective and its kernel, which is nothing but $K \cap g\Fqm[x]$
  is equal to zero.
  Consequently, $K \cap g\Fqm[x]_{<et}$ is zero too.
\end{proof}

\begin{prop}
The $\Fq$--space
$K$ of Definition~\ref{def:K0} satisfies
Conditions~(\ref{item:subset}), (\ref{item:cap}) and (\ref{item:dimK0})
of Proposition~\ref{prop:K}.  
\end{prop}

\begin{proof}
Lemma~\ref{lem:K0} yields Condition~(\ref{item:subset}).
Lemma~\ref{lem:K0} also
yields Condition~(\ref{item:dimK0}) and Corollary~\ref{cor:cap}
gives Condition~(\ref{item:cap}).
\end{proof}

Therefore, we proved Theorem~\ref{thm:irred}~(\ref{item:thmeqbis})
and  there remains to prove Theorem~\ref{thm:irred}~(\ref{item:thmeq}).
Thus, from now on, we assume that $g$ has no roots in $\Fqm$
and we will prove the existence of a one--dimensional
$\Fq$--space $T$ satisfying (\ref{item:K}).
The strategy to find such a $T$ is to choose it as 
$T = \langle \lambda a^{e+1} \rangle_{\Fq}$ for some $a\in \Fqm[x]_{<t}$
and some $\lambda \in \Fqm^{\times}$
satisfying  $\trF{\Fqm}{\Fq} (\lambda) = 0$.
Clearly, we have the following statement.
\begin{lem}\label{lem:in}
For any nonzero element $\lambda\in \Fqm$ such that 
$\trF{\Fqm}{\Fq} (\lambda) = 0$ and for any
$a \in \Fqm[x]_{< t}$, we have
$$
\lambda a^{e+1} \in \ker \tau.
$$
\end{lem}

\begin{proof}
We first observe that for all $a\in \Fqm[x]_{<t}$,
  we have $\lambda a^{e+1} \in \Fqm[t]_{<(e+1)t}$,
  which is elementary. We finish the proof with
  \begin{equation*}
  \TA \left( \evL{\lambda a^{e+1}} \right) = 
  \TA \left( \lambda \cdot \NA( \evL{a}) \right) =
   \trF{\Fqm}{\Fq}(\lambda) \cdot \NA\left( \evL{a} \right) 
   =  0.
  \end{equation*}
\end{proof}


The following proposition is the key to conclude the proof of Theorem~\ref{thm:equality}.

\begin{prop}
  \label{prop:startkey}
  Let $r>1$ be an integer and $\Fqmr$ be the degree $r$ extension
  of $\Fqm$. Let $\lambda \in \Fqm^{\times}$ be such that $\tr_{\Fqm/\Fq}(\lambda)=0$.
  Then, there exists $\alpha \in \Fqmr$ such that 
  $$
  \tr_{\Fqmr/\Fq}(\lambda \alpha^{e+1}) \neq 0. 
  $$
\end{prop}

\begin{proof}
Set 
$
Z\eqdef \left\{ \left. z\in \Fqmr ~\right|~ \tr_{\Fqmr/\Fq}(\lambda z^{e+1}) = 0 \right\}.
$  
Our point is to show that $|Z|<q^{mr}$.
For all $z\in Z$, we have
$$
\begin{aligned}
\tr_{\Fqmr/\Fq}(\lambda z^{e+1}) & = \lambda z^{e+1} + \lambda^q {\left(z^{e+1}\right)}^q +\cdots +\lambda^{q^{mr-1}} {\left(z^{e+1}\right)}^{q^{mr-1}}  \\
   & ={\lambda}{z}^{q^{m-1}+\cdots +1} +  \lambda^q {z}^{q^{m}+\cdots + q}  +
  \cdots  +  {\lambda}^{q^{mr-1}} {z}^{q^{mr + m-2}+\cdots + q^{mr-1}}.
\end{aligned}
$$
  
  \noindent Using the relation ${z}^{q^{mr}} = {z}$, we get
  $$
  \tr_{\Fqmr/\Fq}(\lambda z^{e+1}) = 
   {\lambda} {z}^{R_0(q)} +
  {\lambda}^q {z}^{R_1(q)} +
  \cdots +
  {\lambda}^{q^{mr-1}} {z}^{R_{mr-1}(q)}  
  $$
  where $R_0(q), \ldots , R_{mr-1}(q)$ are integers $<q^{mr}$
  which are sums of $m$ distinct powers of $q$ with exponents
  $<mr$. Namely, 
  \begin{align*}
    R_0(q) & =  q^{m-1} +q^{m-2}  + \cdots  +q +1\\
    R_1(q) & =  q^{m}+q^{m-1} + \cdots +q^2 +q\\
           & \qquad \qquad\vdots   \\
    R_{mr-1}(q) & =  q^{m-2} + \cdots +q +1 + q^{mr-1}.
  \end{align*}
  For all $i$, we have $R_i(q)<q^{mr}$.
  Next, it is not difficult to check that, since by assumption
  $r \geq 2$,
  the $R_i(q)$'s are pairwise distinct since
  they have pairwise distinct $q$--adic expansions.
  Let $Q \in \Fqmr[x]$ be the polynomial
  $$
  Q(x)\eqdef  {\lambda} x^{R_0(q)} +
  {\lambda}^q x^{R_1(q)} +
  \cdots +
  {\lambda}^{q^{mr-1}} x^{R_{mr-1}(q)}.
  $$
  The elements of $Z$ are roots of $Q$
  lying in $\Fqmr$.
  Since the $R_i$'s are pairwise distinct and $\lambda$
  is assumed to be nonzero, the polynomial $Q$ is nonzero.
  In addition, its degree is
  strictly less than $q^{mr}$. Consequently, $Q$ has strictly
  less than $q^{mr}$ roots. Therefore, $|Z|<q^{mr}$, which
  concludes the proof.
\end{proof}

\begin{prop}\label{prop:key}
  Assume that $g$ has no roots in $\Fqm$.
  Let $\lambda $ be a nonzero element of
  $\Fqm$ such that $\trF{\Fqm}{\Fq} (\lambda) = 0$.
  Then, there exists $a\in \Fqm[x]_{<t}$ such that
  \begin{enumerate}[(i)]
  \item\label{item:zero} $K \cap \langle \lambda a^{e+1} \rangle_{\Fq} = \{0\}$;
  \item\label{item:kern} $K \oplus \langle \lambda a^{e+1} \rangle_{\Fq} \subset \ker \tau$;
  \item\label{item:direct} $K \oplus \langle \lambda a^{e+1} \rangle_{\Fq} \oplus g\Fqm [x]_{<et} = \Fqm [x]_{<(e+1)t}$.
  \end{enumerate}
\end{prop}

\begin{proof}
 Recall that $g$ is assumed to be of the form
 $g=h^s$ where $h$ is irreducible and $s\geq 1$.
 The degree of $h$ is denoted by $r$ so that $t=sr$.
 The case $s=1$ corresponds to $g$ irreducible.


 Since $h$ is assumed to have no roots in $\Fqm$,
 we necessarily have $r=\deg(h) \geq 2$.
 Thus, one can apply 
  Proposition~\ref{prop:startkey}, which asserts the existence of
 $\alpha \in \Fqmr \simeq \fract{\Fqm[x]}{(h)}$
 such that $\lambda \alpha^{e+1} \notin \ker \tr_{\Fqmr/\Fq}$.
 From Lemma~\ref{lem:kerTr}, this asserts the existence of $\alpha \in \fract{\Fqm[x]}{(h)}$
 such that 
 \begin{equation}\label{eq:KmodH}
 \lambda \alpha^{e+1} \notin (K \mod (h)).
 \end{equation}
 Let $\alpha_0$ be a lift of $\alpha$ in $\fract{\Fqm[x]}{(g)}$.
 Then, we clearly have
 \begin{equation}
   \label{eq:KmodG}
   \lambda \alpha_0^{e+1} \notin (K \mod (g)).
 \end{equation}
 Indeed, if we had $\lambda \alpha_0^{e+1} \in (K \mod (g))$, then
 reducing modulo $(h)$ we would contradict~(\ref{eq:KmodH}).
 From Proposition~\ref{prop:dimKbar}, we know that $K \mod (g)$
 has $\Fq$-codimension $1$ in $\fract{\Fqm[x]}{(g)}$ and hence
 (\ref{eq:KmodG})
 yields
 \begin{equation}\label{eq:Modmod}
 (K \mod (g)) \oplus 
 \langle \lambda \alpha_0^{e+1}\rangle_{\Fq} = \fract{\Fqm[x]}{(g)}.
 \end{equation}
 Now, let $a\in \Fqm[x]_{<t}$ be a lift of $\alpha_0$. Here again, we clearly
 have
 $\lambda a^{e+1}\notin K$. This proves (\ref{item:zero}).
 Afterwards, (\ref{item:kern}) is a direct consequence of Lemma~\ref{lem:in}.

 Finally, from Lemma \ref{lem:K0}, we have
 \begin{equation}\label{eq:pouetpouet}
 \dim_{\Fq} K \oplus \langle \lambda a^{e+1}\rangle_{\Fq} = mt.
 \end{equation}
 Since $mt$ is nothing but the $\Fq$--dimension of $\fract{\Fqm[x]}{(g)}$,
 we see that (\ref{eq:Modmod}) together with (\ref{eq:pouetpouet}) prove
 that the space $K \oplus \langle \lambda a^{q+1}\rangle_{\Fq}$
 is isomorphic to its reduction modulo $(g)$, which entails:
 $$
 (K \oplus \langle \lambda a^{e+1}\rangle_{\Fq}) \cap g\Fqm[x] = \{0\}.
 $$
 This leads to (\ref{item:direct}) and
 terminates the proof.
\end{proof}

\noindent {\bf Conclusion.} Proposition~\ref{prop:key} gives the existence
of a vector space $T$ satisfying Condition (\ref{item:K})
of Proposition~\ref{prop:K}. This concludes the proof of Theorem
\ref{thm:irred}~(\ref{item:thmeq}) and hence of Theorem~\ref{thm:equality}.

\begin{rem}
  It is worthwhile 
 noting that the condition 
  ``$g$ has no roots in $\Fqm$'' is necessary to prove
  the result since it is necessary to prove Proposition~\ref{prop:startkey}.
  Indeed, it is easy to see that if $g$ had roots in $\Fqm$,
  then we would need to prove Proposition~\ref{prop:startkey} 
  for $r=1$.
  However, the proof of Proposition~\ref{prop:startkey}
  does not hold for $r=1$, since in that case
  all the $R_i(q)$'s in the proof
  would be equal and the polynomial
  $Q$ would be zero. 
  Using {\sc Magma}
  \cite{magma}, it is easy to compute examples of
  Goppa codes $\gop{L}{g^{e}}$ and $\gop{L}{g^{e+1}}$, which are distinct
  when $g$ has roots in $\Fqm$.
  Thus, one cannot expect better than Theorem~\ref{thm:eqbis}.
  This is illustrated by the examples in \S~\ref{subsec:further}.
\end{rem}


\section{Examples}
\label{sec:examples}

In this section we consider some specific situations to illustrate our results.
We first focus on the case of quadratic extensions, that is to say $m=2$.
Next, we illustrate Theorem \ref{thm:irred} by considering such codes
with a polynomial $g$ of degree $1$ and an extension degree 
that is equal to $3$.

\subsection{Wild Goppa codes from quadratic extensions}
In this subsection, the extension degree $m$ will be equal to $2$.
In this particular situation, our Theorem~\ref{thm:equality} asserts that
for a squarefree polynomial $g$ with no roots in $\Fqq$, we have
$$
\gop{L}{g^{q-1}} = \gop{L}{g^q} = \gop{L}{g^{q+1}}.
$$
Therefore, the minimum distance of this code is bounded below
by $\deg(g)(q+1)+1$ instead of $\deg(g)q+1$, which was its designed distance 
up to now. Another striking fact is that its dimension is also larger than
the lower bound $n-2\deg(g)(q-1)$.
This is a consequence of the results of \S~\ref{sec:proof1}, which assert
that such a code 
is diagonally equivalent to a subfield subcode of a Reed--Solomon code or a shortening of it.
The following statement yields a lower bound for the dimension of these
wild Goppa codes.
\begin{thm}
  \label{thm:ttmoins2}
  Let $g\in \Fqq[x]$ be a polynomial of degree
  $t\geq 2$ with no roots in
  $\Fqq$ and $L$ be a support of length $n$, then
  $$
  \dim_{\Fq} \gop{L}{g^{q+1}} \geq n - 2 t (q+1) + t(t+2).
  $$
  In addition, the inequality is an equality when $L$ is a full-support
  or a support of length $q^2-1$.
\end{thm}

\begin{proof}
  First, let us assume that $L$ is a full support, {\it i.e.} $n=q^2$.
  From Corollary~\ref{cor:equivalence}, the Goppa code $\gop{L}{g^{q+1}}$
  is $\Fq$--equivalent to the subfield subcode of $RS_{q^2-t(q+1)}(L)$.
  The dimension of such a code is boun\-ded below in \cite{HMO'S}.
  This bound concerns codes supported by $\Fqm\setminus \{0\}$ and its
  shortenings. However, the case of a full support can easily be deduced
  from that of the support $\Fqm \setminus\{0\}$, since the latter
  is nothing but the shortening of the former at one position.

  Before stating this lower bound, let us recall some notions
  and notation
  on cyclotomic classes.
  We call a {\em cyclotomic class} an orbit of $\fract{\Z}{(q^2-1)\Z}$
  for the multiplication by $q$.
  One sees easily that, in this situation, cyclotomic classes
  contain either one or two elements. For instance $\{0\}$, $\{q+1\}$
  or $\{1,q\}$ are cyclotomic classes.
  From now, on, we denote by $B$ the set of smallest elements in the
  cyclotomic classes. For all $b\in B$, we denote by $I_b$ the corresponding
  class and by $n_b$ the cardinality of $I_b$. In addition, we denote by $A$ the set
  $\{0, \ldots , t(q+1)-1\}$.
  From \cite[Theorem 4.4]{HMO'S} (applied to $m=2$), we have
  \begin{equation}\label{eq:lowb}
  \dim_{\Fq}(RS_{q^2-t(q+1)}(L))_{|\F_q} = q^2 - 2t(q+1) + \sum_{b \in B\cap A}
  (2(|I_b \cap A| -1)+2-n_b).
  \end{equation}
  Actually, \cite[Theorem 4.4]{HMO'S} is an inequality, 
  but below this statement in \cite{HMO'S},
  the equality cases are discussed and equality holds
  always for a full support.

  The sum in~(\ref{eq:lowb}) involves two kinds of cyclotomic classes,
  namely:
  \begin{itemize}
    \item the classes $I_b$ with $I_b \subset A$ and $n_b=1$.
      These classes are $\{0\}, \{q+1\}, \ldots , \{(t-1)(q+1)\}$.
      Their number is equal to $t$.
    \item the classes $I_b$ with $I_b \subset A$ and $n_b =2$.
      These classes are of the form $\{a_0+a_1q, a_0q+a_1\}$ for 
      $(a_0, a_1) \in \{0, \ldots, t\}^2$ and $a_0 \neq a_1$.
      The number of such classes is ${t+1 \choose 2}$.
  \end{itemize}
  It is easy to observe that the other cyclotomic classes have no contribution
  in the sum in~(\ref{eq:lowb}). Consequently, we get
  $$
  \begin{aligned}
    \dim_{\Fq}(RS_{q^2-t(q+1)}(L))_{|\Fq} & = q^2 - 2t(q+1) + 2{t+1 \choose 2} + t\\
       & =q^2 -2t(q+1) +t(t+2).
  \end{aligned}
  $$
  
  Now, if $L$ is an arbitrary support, then, from
  Lemma
  \ref{lem:supports}, the code $\gop{L}{g^{q+1}}$
  is the shortening of a full support Goppa code.
  Hence it is $\Fq$--equivalent to the shortening
  of $RS_{q^2-r(q+1)}(L_0)_{|\Fq}$, where $L_0$ denotes a full support.
  Therefore, the
  general case results straightforwardly
  from the full support case.
\end{proof}

\begin{rem}
  If we reconsider the wild Goppa code $\gop{L}{g^{q-1}}$ whose designed
dimension is $n - 2 t (q-1)$, we see that if $t\geq 3$ then the actual
dimension is larger and the difference between the actual and the designed 
dimension is $t(t-2)$. It is quadratic in $t$.
\end{rem}

Table~\ref{tab:wild:m=2} lists the parameters of some of these codes.
It turns out that all these parameters reach those
of the best known codes listed in~\cite{grassl}.
\begin{table}[h] 
\centering
\begin{tabular}{@{}ccccc@{}}
  \toprule
    & $q= 5$ & $q= 7$ & $q= 8$ & $q= 9$ \\
  \midrule
  $\deg(g) = 3$ & [25 ,  4 , $\geq 19$] & [49,  16, $\geq 25$]  & [64,  25, $\geq 28$]  &  [81,  36, $\geq 31$]  \\
  $\deg(g) = 4$ & - & [49,  9, $\geq 33$]  & [64,  16, $\geq  37$]  &  [81,  25, $\geq 41$]\\
  $\deg(g) = 5$ & - & [49,  4, $\geq 41$ ]  & [64,  9, $\geq 46$]  &  [81,  16, $\geq 51$]\\
  $\deg(g) = 6$ & - & -  & [64,  4, $\geq 55$]  &  [81,  9, $\geq 61$]\\
  $\deg(g) = 7$ & - & -  & - &  [81,  4, $\geq 71$]\\
  \bottomrule
\end{tabular}
\caption{Parameters of Wild Goppa codes over 
a 
quadratic extension ($m =2$).} \label{tab:wild:m=2}
\end{table}

\subsection{Further examples}\label{subsec:further}
Now, let us consider the case of cubic extensions, that is $m=3$ and
the particular case of a polynomial $g$ of degree $1$.
First, let us state a general result on the dimension
of such codes from cubic extensions.
\begin{thm}\label{thm:dimcub}
  Let $g\in \F_{q^3}[x]$ be a polynomial of degree $t$ and $L \in \F_{q^3}^n$
  be a support of length $n$ avoiding the roots of $g$. Then,
  $$
  \dim_{\Fq} \gop{L}{g^{q^2+q+1}} \geq n - 3t(q^2+q+1) +2t+2t(t+1)(t+2)
  +3(q-1-t)t(t+1)
  $$
  and equality holds if $L$ is a full support or has length $q^3-1$.
\end{thm}

\begin{proof}
  We use the very same techniques as in the proof of
  Theorem~\ref{thm:ttmoins2}
  and use the same notation with the only change that
  here cyclotomic classes are subsets of $\fract{\Z}{(q^3-1)\Z}$
  and $A = \{0, \ldots , t(q^2+q+1)-1\}$.
  Here \cite[Theorem 4.4]{HMO'S} asserts
  that 
  \begin{equation}
    \label{eq:cubiccase}
  \dim_{\Fq} \gop{L}{g^{q^2+q+1}} \geq n - 3t(q^2+q+1) + \sum_{b\in B\cap A}
  (m(|I_b \cap A|-1)+m-n_b).
  \end{equation}
  We consider three kinds of cyclotomic classes.
  \begin{itemize}
  \item The classes $\{0\}, \{q^2+q+1\}, \ldots , \{(t-1)(q^2+q+1)\}$.
    Their number is $t$, they satisfy $n_b=1$ and $|I_b \cap A|=1$.
    They yield a term $2t$
    in the sum in the second member of~(\ref{eq:cubiccase}).
  \item The classes $\left\{\{a_0 + a_1 q + a_2 q^2\},\{a_2 + a_0 q + a_1 q^2\},
      \{a_1 + a_2 q + a_0 q^2\}
      \right\}$ for $a_i \leq t$
    and at least one of the $a_i$'s is distinct from the others.
    They satisfy $n_b = 3$, $|I_b \cap A| = 3$ and their number is
    $\frac{(t+1)^3-(t+1)}{3}$. They provide a term $2t(t+1)(t+2)$
    in the sum in the second member of~(\ref{eq:cubiccase}).
  \item The classes $\left\{\{a_0 + a_1 q + a_2 q^2\},\{a_2 + a_0 q + a_1 q^2\},
      \{a_1 + a_2 q + a_0 q^2\} \right\}$
    for $t<a_2\leq q-1$ and $0\leq a_0 < t$ and $0\leq a_1 \leq t$.
    They satisfy $n_b = 3$ and $|I_b \cap A| = 2$. Their number is
    $(q-1-t)((t+1)^2-(t+1))$ and they provide a term 
    $3(q-1-t)t(t+1)$
    in the sum in the second member of~(\ref{eq:cubiccase}).
  \end{itemize}
  It can be checked that no other cyclotomic class contributes in the sum
  in~(\ref{eq:cubiccase}) and combining the three above items, we get the
  result.
\end{proof}

Now, let us focus on the case of a polynomial $g$ of degree $1$.
For the support $L$ we take a vector of length $q^3-1$ listing
every element of $\F_{q^3}$ but the single root
of $g$.
Here, Theorem~\ref{thm:dimcub} gives
\begin{equation}\label{eq:norm}
\dim_{\Fq} \gop{L}{g^{q^2+q+1}} \geq (q^3-1) - 3 (q^2+q+1) +14 +6(q-2).
\end{equation}
On the other hand, the classical bound for alternant codes yields
\begin{equation}\label{eq:normdivg}
\dim_{\Fq} \gop{L}{g^{q^2+q}} \geq (q^3-1) - 3 (q^2+q).
\end{equation}
Obviously, since we have the inclusion $\gop{L}{g^{q^2+q+1}} \subset
\gop{L}{g^{q^2+q}}$, and comparing the bounds, we see that (\ref{eq:normdivg})
is far from being sharp and that (\ref{eq:norm}) gives a better lower bound
for the dimension of the code $\gop{L}{g^{q^2+q}}$.
In addition, Theorem~\ref{thm:eqbis} asserts that $\gop{L}{g^{q^2+q}}$
might have one dimension more than $\gop{L}{g^{q^2+q+1}}$.
This is what happens in general.
In Table~\ref{tab:g=x},
we give the parameters of such Goppa codes when the polynomial $g$ is $x$.
The true dimensions have been verified with {\sc magma} \cite{magma}.
They coincide with the above discussed lower bounds.
\begin{table}[h]
\centering
\begin{tabular}{@{}ccccc@{}}
\toprule
 & $q=4$ & $q =5$ & $q=7$ & $q=8$  \\
\midrule
$\gop{L}{x^{q^2+q+1}}$ & [63, 26, $\geq 22$] & [124, 63, $\geq 32$] & 
[342, 215, $\geq 58$] &  [511, 342, $\geq 74$] \\  
$\gop{L}{x^{q^2+q}}$ & [63, 27, $\geq 21$] & [124, 64, $\geq 31$] &  
[342, 216, $\geq 57$] &  [511, 343, $\geq 73$] \\
\bottomrule
\end{tabular}
\caption{Parameters of wild Goppa codes with $g = x$ and $m=3$.} \label{tab:g=x}
\end{table}

\section*{Conclusion}
We proved two new identities relating so--called wild Goppa codes.
The first one asserts that if $g$ is a polynomial with no
roots in $\Fqm$, then $\gop{L}{g^{q^{m-1}+\cdots +q^2+q}} =
\gop{L}{g^{q^{m-1}+\cdots +q^2+q+1}}$.
The second one asserts that if $g$ has roots in $\Fqm$ then, the equality
fails but the difference of the $\Fq$--dimensions of the two codes is bounded
above by the number of distinct roots of $g$ in $\Fqm$.
The corresponding codes are of particular interest since they turn out
to be extended or shortened BCH codes and have a very high 
dimension compared to the designed dimension of alternant codes.

It should be pointed out that the proofs of our main results
in the present article involve duals of Goppa codes.
Getting  direct proofs of such identities involving only the rational
fractions used to define Goppa codes would be of interest.

\section*{Acknowledgements}
 The authors express their deep gratitude to Sergey Bezzateev for his
 careful reading and its relevant comments on this work.

\bibliographystyle{abbrv}
\bibliography{biblio}

\end{document}